\documentclass[12pt,reqno]{amsart}
\pdfoutput=1 
\usepackage{amsmath,bbm}
\usepackage{latexsym}
\usepackage{amsfonts}
\usepackage{amssymb}
\usepackage{color}
\usepackage{graphicx}
\usepackage{url}
\usepackage{enumerate}
\usepackage{tikz}
\usetikzlibrary{shadings, intersections, calc, plotmarks}
\usepackage{marginnote}
\usepackage{stackrel}
\usepackage{empheq}
\usepackage{tcolorbox}
\usetikzlibrary{arrows.meta,positioning}

\usepackage{geometry}
\usepackage{hyperref}
\hypersetup{
  colorlinks=true,
  linkcolor=blue!90!black,    
  citecolor=green!50!black,   
  urlcolor=blue!70!black,     
  linktoc=all,                
}

\xdefinecolor{tumblue}     {RGB}{0,101,189}
\xdefinecolor{tumgreen}    {RGB}{162,173,  0}
\xdefinecolor{tumred}      {RGB}{229, 52, 24}
\xdefinecolor{tumivory}    {RGB}{218,215,203}
\xdefinecolor{tumorange}   {RGB}{227,114, 34}
\xdefinecolor{tumlightblue}{RGB}{152,198,234}

\newtheorem{proposition}{Proposition}
\newtheorem*{proposition*}{Proposition}

\newtheorem{theorem}{Theorem}
\newtheorem*{theorem*}{Theorem}
\newtheorem{lemma}{Lemma}
\newtheorem{corollary}{Corollary}
\newtheorem*{corollary*}{Corollary}
\newtheorem{conjecture}{Conjecture}

\newtheorem*{definition*}{Definition}

\newcommand{\spec}{{\rm{spec}}}

\newcommand{\cH}{\mathcal{H}}
\newcommand{\R}{\mathbbm{R}}

\newcommand{\cK}{\mathcal{K}}

\newcommand{\C}{\mathbbm{C}}
\newcommand{\cF}{\mathcal{F}}
\newcommand{\N}{\mathbbm{N}}

\newcommand{\F}{\mathbbm{F}}

\newcommand{\cP}{\mathcal{P}}
\newcommand{\cT}{\mathcal{T}}

\newcommand{\cA}{\mathcal{A}}
\newcommand{\cE}{\mathcal{E}}

\newcommand{\cV}{\mathcal{V}}

\newcommand{\cY}{{\mathcal Y}}

\newcommand{\1}{\mathbbm{1}}
\newcommand{\Real}{\mathfrak{Re}}

\def\>{{\rangle}}
\def\<{{\langle}}
\def\Z{{\mathbb{Z}}}
\newcommand{\be}{\begin{equation}}
\newcommand{\ee}{\end{equation}}
\newcommand{\bea}{\begin{eqnarray}}
\newcommand{\eea}{\end{eqnarray}}

\newcommand{\tr}[1]{\mathrm{tr}\left[#1\right]} 

\newcommand{\rank}[1]{\mathrm{rank}\left[#1\right]}
\newcommand{\norm}[1]{\left\lVert #1 \right\rVert}

\renewcommand{\d}[1]{\ensuremath{\operatorname{d}\!{#1}}}

\newcommand{\abs}[1]{\left\lvert #1 \right\rvert}

\newcommand\rb[1]{\left( #1 \right)}

\newcommand\dotp[1]{\left< #1 \right>}

\newcommand\suml[2]{\sum\limits_{#1}^{#2}}

\renewcommand{\vec}[1]{\left|{#1}\right>}
\newcommand{\dvec}[1]{\left<{#1}\right|}
\newcommand{\op}[2]{\vec{#1}\dvec{#2}}

\begin{document}

\title[RA conjecture \& equipartitions of positive operators]{The Ruskai-Audenaert conjecture \&\\ equipartitions of positive operators}

\author[Kumar]{Niranjan Kumar$^{1,2,3}$}
\email{niranjan.kumar@mpq.mpg.de}
\author[Wolf]{Michael M. Wolf$^{1,2}$}
\email{m.wolf@tum.de}
\address{$^1$ Department of Mathematics, Technical University of Munich}
\address{$^2$ Munich Center for Quantum
Science and Technology (MCQST),  Munich, Germany}
\address{$^3$ Max-Planck-Institute of Quantum Optics, Garching, Germany}

\date{\today}

\begin{abstract} Several open problems in quantum information theory can be formulated as equipartition problems for positive operators, asking for a decomposition into bounded-rank positive parts under uniform constraints. The existence problems for SIC-POVMs and MUBs are of this type, as is the Ruskai–Audenaert conjecture. We first show that certain problems of this form can be attacked using equivariant cohomology, and then present new results on the Ruskai–Audenaert conjecture. In its weak form, this conjecture asserts that every quantum channel admits a convex decomposition into a minimal number of generalized extreme points; in its strong form, one with equal weights. We prove the strong conjecture in all dimensions for a set of channels of nonzero measure, including all cq- and qc-channels, as well as for all channels with qubit inputs. We also prove the weak conjecture for all qutrit channels, along with further results on convex decompositions of quantum channels.
\end{abstract}
\maketitle
\section{Introduction}\label{sec:intro}
A \emph{quantum channel} mathematically describes the input-output relation of a quantum process---such as the transformation of a (possibly open) quantum system, the transmission through a communication line, or the evolution of a quantum computer. In the Schr\"odinger picture, a quantum channel is a completely positive, trace-preserving (cptp) linear map. If we fix the Hilbert space dimensions of the input and output systems to $d_1$ and $d_2$, respectively, the set of all quantum channels, which then map $\C^{d_1\times d_1}\rightarrow \C^{d_2\times d_2}$, is convex and compact. By Carath\'eodory's theorem, every element of this set can therefore be written as a convex combination of at most $D+1$ extreme points, where $D$ is the dimension of the ambient real affine space. Here, that space is the space of Hermiticity- and trace-preserving linear maps, so $D=(d_2^2-1)d_1^2$. 

The topological closure $\overline{\cE}$ of the set of extreme points $\cE$ turns out to be the set of all quantum channels whose `Kraus rank', which is the rank of the Choi matrix, is at most $d_1$. Elements of  $\overline{\cE}$ are called \emph{generalized extreme points}. Since extremality can be characterized in terms of a linear independence condition \cite{CHOICP}, the complement of $\cE$ in  $\overline{\cE}$ is of lower dimension, so that, generically, a generalized extreme point \emph{is} an extreme point.  

In \cite{Ruskai} M.B. Ruskai formulated a remarkable conjecture, based on numerical evidence provided by K. Audenaert. We refer to it as the \emph{Ruskai-Audenaert  conjecture} (RA) and distinguish a weak and strong form of RA: in its weak form, RA asserts that every quantum channel admits a convex decomposition into only $d_2$ generalized extreme points. The strong form states that additionally the corresponding weights are all equal. For instance, for $d_1=d_2=3$, RA yields a convex decomposition into $3$ terms (with equal weights in the strong form), whereas Carath\'eodory merely guarantees one into $73$ terms with different weights.\vspace{5pt}

In this paper, we will approach the Ruskai-Audenaert conjecture from three different directions: (i) via a larger class of equipartition problems of positive operators that can be addressed using equivariant topology (Sec.\ref{sec:balanced}), (ii) via a general analysis of convex decompositions of quantum channels and their structural robustness (Sec.\ref{sec:convex}), and (iii) via explicit classes of channels for which strong RA can be proven directly (Sec.\ref{sec:strongRAexp}).\vspace*{5pt}

More specifically:

\begin{itemize}
 \item[$\circ$] In Sec.~\ref{sec:balanced} (and Appendices~\ref{sec:SICPOVM}, \ref{sec:MUB}) we show that the existence problems for SIC-POVMs and MUBs, as well as the RA conjectures, can be phrased as equipartition problems for positive operators. Thm.~\ref{thm:equivP} then provides a general tool, based on equivariant cohomology, for solving some problems of this kind. Its corollaries are a stronger RA decomposition for $d_1=d_2=2$ (Cor.~\ref{cor:strongerRA}), a proof of the weak RA conjecture for $d_1=d_2=3$ (Cor.~\ref{cor:weak3RA}), a stronger Schur--Horn decomposition for density operators (Cor.~\ref{cor:SHext}) and one with similar marginals (Cor.~\ref{cor:marg}), and the existence of orthonormal bases with balanced prescribed properties (Prop.~\ref{prop:balONB}).\vspace*{5pt}
\item[$\circ$] In Sec.~\ref{sec:convex} we focus on convex decompositions of quantum channels. Thm.~\ref{thm:robustness} provides a condition for the structural robustness of such a decomposition in terms of the connectedness of a graph assigned to it. Exploiting this, Thm.~\ref{thm:RAstructure} then shows that the set of channels for which the strong RA conjecture holds is closed, connected, semialgebraic and has nonempty interior, i.e., nonzero measure in every dimension. After stating an improved upper bound on the Carath\'eodory number (Thm.~\ref{thm:Cnr}), Thm.~\ref{thm:bary} shows that barycentric decompositions of quantum channels into generalized extreme points always exist, whereas such decompositions into extreme points need not exist if $d_2\geq 3$.\vspace*{5pt}
\item Sec.~\ref{sec:strongRAexp} provides explicit families of quantum channels with a proof of the strong RA conjecture. Thm.~\ref{thm:strong2RA} establishes this for any $d_2$ if $d_1=2$, and Thm.~\ref{thm:cqqc} proves the strong RA conjecture for all qc- and cq-channels.
\end{itemize}

\newpage
\section{Equipartitions of positive operators}\label{sec:balanced}

\noindent Our template for an `equipartition problem' of a positive operator is the following:

\begin{quote}
\begin{tcolorbox}[
  colback=gray!10,    
  colframe=gray!10,   
  boxrule=0pt,        
  left=.51em, right=.51em, top=.61em, bottom=0.2em
]
Given a positive operator $P$ on a complex Hilbert space, and $n,r\in\N$, is there a decomposition $$P=\sum_{i=1}^n P_i$$ into positive operators $(P_i)_{i=1}^n$, each of rank at most $r$ and  satisfying $f(P_1)=\cdots=f(P_n)=c$ for a given continuous map $f$?
\end{tcolorbox}
\end{quote}
Note that for a linear constraint function $f$, the constant $c$ is already determined by $f(P_i)=\tfrac1n f(P)$.

A number of open problems in quantum information theory can be expressed as equipartition problems of positive operators. These include notorious ones, such as the existence problems of \emph{symmetric informationally complete POVMs} (SIC-POVMs) and of \emph{mutually unbiased bases} (MUBs) as well as the \emph{Ruskai-Audenaert decomposition} problem. 

We begin by formulating the existence problems for  SIC-POVMs and MUBs as equipartition problems of positive operators. References and proofs for these cases are included in Appendices \ref{sec:SICPOVM} and \ref{sec:MUB}.

\begin{proposition}[SIC-POVMs]\label{prop:SIC} Let $P_{sym}$ be the Hermitian projector onto the symmetric subspace of $\C^d\otimes\C^d$. There exists a SIC-POVM in $\C^d$ iff there are positive rank-one operators $P_1,\ldots,P_{d^2}$ s.t. $\frac{2d}{d+1}P_{sym}=\sum_{i} P_i$ and for all $i$: $\tr{{\rm tr}_1[P_i]^2}=\tr{P_i}^2$.  
\end{proposition}
\begin{proposition}[MUBs]\label{prop:MUB} Let $\omega$ be the projector onto a maximally entangled state in $\C^d\otimes\C^d$ and $P:=\1-\omega$.
    There are $(d+1)$ MUBs in $\C^d$ iff there are positive operators $(P_i)_{i=1}^{d+1}$ of rank $(d-1)$ s.t. $P=\sum_{i} P_i$ with partial traces ${\rm tr}_1[P_i]={\rm tr}_2[P_i]\propto \1$ and partial transposes s.t. $\|(P_i+\omega)^{T_1}\|_1=d$. 
\end{proposition}

Now we turn towards our main target, formulated as  equipartition problems of positive operators:

\begin{conjecture}[Ruskai-Audenaert]\label{conj:RA1}
    Let $P\in {\rm Herm}(\C^{d_1}\otimes \C^{d_2})$ be the Choi matrix of a completely positive trace-preserving map (cptp) $T:\C^{d_1\times d_1}\rightarrow \C^{d_2\times d_2}$. There are $n=d_2$ positive operators $P_i$ of rank at most $r= d_1$ s.t. $P=\sum_{i=1}^n P_i $ and either
    \begin{itemize}
        \item ${\rm tr}_2[P_i]\propto\1$ for all $i$, or \hfill (weak RA conjecture)
        \item ${\rm tr}_2[P_i]={\rm tr}_2[P_j]$ for all $i,j$.\hfill (strong RA conjecture)
    \end{itemize}
\end{conjecture}
The \emph{weak RA conjecture} asserts that within the set of cptp maps,  every $T$ can be convexly decomposed into $d_2$ generalized extreme points. Note that repeated use of generalized extreme points is allowed. Moreover, $d_2$ terms in the decomposition is the absolute minimum since a Choi matrix of rank $d_1 d_2$ cannot be a sum of fewer than $d_2$ matrices of rank $d_1$ or less.

The \emph{strong RA conjecture} states that, in addition to the weak conjecture, the weights in the convex decomposition can be chosen to be all equal, i.e., the decomposition is \emph{barycentric}.\vspace{5pt} 

Another example that has the structure of the template problem is one that can be proven using the Schur-Horn theorem (cf. \cite{Ruskai}). It can be regarded as the $d_1=1$ case of the strong RA conjecture:
\begin{proposition}[Schur-Horn decomposition]\label{prop:SchurHorn} For any density operator $\rho$ acting on $\C^n$, there are $\varphi_i\in\C^n$ with $\|\varphi_i\|=1$ s.t.
    \begin{equation}
        \rho=\frac1n\sum_{i=1}^n |\varphi_i\rangle\langle\varphi_i|.
    \end{equation}
\end{proposition}

\subsection*{Equipartitions from equivariant cohomology}
A simple example that shows that and how topological tools can be helpful in proving the existence of an  equipartition, is the following: 
\subsection*{A toy example} Consider $\cH=\C^2$, $r=1$, $P>0$, and a continuous map $f:{\rm Herm}(\cH)\rightarrow\R^2$. Our goal is to show that there are always two positive rank-one operators $P_1, P_2$ s.t. $P=P_1+P_2$ and $f(P_1)=f(P_2)$. 

To this end, note that the space of all decompositions of $P$ into two positive rank-one operators is homeomorphic to the sphere $S^2$. This can be seen by writing $P_i=\sqrt{P}Q_i\sqrt{P}$ and realizing that $Q_1$ and $Q_2$ must be mutually orthogonal projectors, which in turn correspond to antipodal points on the \emph{Bloch sphere} $S^2$. So write $P_x,P_{-x}$ with $x\in S^2$ instead of $P_1, P_2$. Now define $\psi:S^2\rightarrow \R^2$, $\psi(x):=f(P_x)-f(P_{-x})$. As this is an \emph{odd} map, in the sense that $\psi(-x)=-\psi(x)$, the \emph{Borsuk-Ulam} theorem guarantees the existence of a zero and thus of an $x$ s.t. $f(P_x)=f(P_{-x})$, while $P=P_x+P_{-x}$. \vspace*{6pt}   

Now we will extend the idea behind this toy-example to more general scenarios by employing standard ideas from equivariant topology (cf. \cite{Equivariant}, \cite{Jelic}).
Denote by $\cF_r(\C^{rn})$ the set of all $n$-tuples $(Q_1,\ldots,Q_n)$ of mutually orthogonal projectors $Q_i\in {\rm Herm}(\C^{rn})$ of rank $r$. Note the homeomorphism $\cF_r(\C^{rn})\simeq U(rn)/U(r)^n$ to a \emph{flag manifold}  and that $\sum_i Q_i=\1$. We can regard $\cF_r(\C^{rn})$ as equipped with an action of the group $S_n$ of permutations of $n$ elements,  simply by  permuting the $n$ entries of each $n$-tuple. That is, $\sigma\cdot (Q_1,\ldots,Q_n):=(Q_{\sigma^{-1}(1)},\ldots,Q_{\sigma^{-1}(n)})$ for any $\sigma\in S_n$.

We also define a \emph{test space} $W_n(\R^k):=\{x\in(\R^k)^n | x_1+\ldots +x_n=0\}$, again equipped with the natural action of $S_n$, and its unit sphere $S\big(W_n(\R^k)\big)$ with respect to the Euclidean norm in $\R^{kn}$.

A continuous map $\Psi:\cF_r(\C^{rn})\rightarrow S\big(W_n(\R^k)\big)$ is called $G$-\emph{equivariant} for any subgroup $G\subseteq S_n$ if $\sigma\circ\Psi=\Psi\circ\sigma$ for all $\sigma \in G$. 

\begin{lemma}\label{lem:balancedequi}
    Given an operator $P\geq 0$ on a complex finite-dimensional Hilbert space $\cH$  and a continuous map $f:{\rm Herm}(\cH)\rightarrow \R^k$. Let $n,r\in\N$ be such that $\rank{P}\leq rn$. If there is no $S_n$-equivariant map $\cF_r(\C^{rn})\rightarrow S\big(W_n(\R^k)\big)$, then there are positive operators $(P_i)_{i=1}^n$ of rank at most $r$ s.t.  $P=\sum_{i=1}^n P_i$ and $f(P_1)=\cdots=f(P_n)$. 
\end{lemma}
\emph{Remark: } Note that we could replace $S_n$ with any subgroup $G\subseteq S_n$ since every $S_n$-equivariant map is in particular $G$-equivariant. However, this also means that $G=S_n$ gives the logically strongest implication. 
\begin{proof} 
We will construct an $S_n$-equivariant map under the assumption that the described decomposition of $P$ does not exist.

Denote by $\cH_P:={\rm supp}(P)$ the support space of $P$ and decompose $\cH=\cH_P\oplus\cH_P^\perp$. Since $\dim(\cH_P)\leq rn$, we can define a linear map $V:\cH\rightarrow\C^{rn}$ that acts isometrically on $\cH_P$ and maps $\cH_P^\perp$ to $\{0\}$.

    For any $Q\in\cF_r(\C^{rn})$, define $P_i:=\sqrt{P}V^*Q_iV\sqrt{P}$. Since $P$ and $Q_i$ are positive, each $P_i$ is positive, of rank at most $r$, and $\sum_i P_i =P$. Define $\psi(Q):=\big(f(P_1)-\overline{f},\ldots,f(P_n)-\overline{f}\big)$, where $\overline{f}:=\tfrac{1}{n}\sum_i f(P_i)$. By construction, $\psi(Q)\in \big(W_n(\R^k)\big)$ and $\psi(Q)=0$ iff $f(P_1)=\cdots =f(P_n)$. Assuming that a decomposition with this property does not exist, we can define a $S_n$-equivariant map $\Psi:\cF_r(\C^{rn})\rightarrow S\big(W_n(\R^k)\big)$ via $\Psi(Q):=\psi(Q)/\|\psi(Q)\|$.
  
\end{proof}

A standard tool for disproving the existence of a $G$-equivariant map $\Psi:X\rightarrow Y$ between two topological spaces with a $G$-action is the \emph{Fadell-Husseini index} \cite{FHindex}, which we will now briefly review. Suppose a $G$-equivariant $\Psi$ exists, and denote by $EG$ a contractible space on which $G$ acts \emph{freely}\footnote{$G$ acts \emph{freely} on $X$ if for every $g\in G$ that is not the identity and any $x\in X: g\cdot x\neq x$. In \cite{Milnor56} Milnor provided a construction of a contractible space $EG$ with free group action for any topological group.}. Then one can construct a commuting diagram (left part of the subsequent figure) using the \emph{Borel construction} $X_G:=(EG\times X)/G$, where $\pi_X:[e,x]\mapsto [e]$, $X_G\rightarrow Y_G: [e,x]\mapsto [e,\Psi(x)]$, and all arrows correspond to $G-$equivariant maps:

\[
\begin{tikzpicture}[>=Stealth, baseline=(d2.base)]
\hspace*{-5pt}

  \node (b1) at (4.5,0) {$\bigl(EG\times X\bigr)/G$};
  \node (b2) at (9.5,0) {$\bigl(EG\times Y\bigr)/G$};
  \node (b3) at (7.0,-1.8) {$EG/G$};
  \draw[->] (b1) -- (b2);
  \draw[->] (b1) -- node[left] {$\pi_X$} (b3);
  \draw[->] (b2) -- node[right] {$\pi_Y$} (b3);

  \node (c1) at (12.3,0) {$H^{*}(X_G;R)$};
  \node (c2) at (16.5,0) {$H^{*}(Y_G;R)$};
  \node (d2) at (14.4,-1.8) {$H^{*}(EG/G;R)$};
  \draw[->] (c2) -- (c1);
  \draw[->] (d2) -- node[left] {$\pi_X^{*}$} (c1);
  \draw[->] (d2) -- node[right] {$\pi_Y^{*}$} (c2);
\end{tikzpicture}
\]
Applying the contravariant functor of singular cohomology with coefficients in a commutative ring $R$ with unit to the diagram on the left, one obtains the one on the right, where the arrows are reversed. Commutativity of the diagram then implies that 
\begin{equation}\label{eq:FHinclusion}
    \ker\big(\pi^*_Y\big)\ \subseteq\ \ker\big(\pi^*_X\big). 
\end{equation}
These kernels, which are ideals in the cohomology ring $H^*(EG/G;R)$, are called the \emph{Fadell-Husseini indices} of the spaces $Y$ and $X$. The non-existence of an equivariant map can thus be proven by showing that the inclusion in Eq.(\ref{eq:FHinclusion}) fails. Following Lemma \ref{lem:balancedequi}, this yields a clear path towards proving the existence of an equipartition, which results in the following:

\begin{theorem}\label{thm:equivP}
    Given an operator $P\geq 0$ on a finite-dimensional complex Hilbert space $\cH$ and a continuous map $f:{\rm Herm}(\cH)\rightarrow \R^k$. Let $n\in\N$ be prime and $r\in\N$ s.t. $\rank{P}\leq rn$. There are positive operators $(P_i)_{i=1}^n$ of rank at most $r$ such that \begin{equation}
    P=\sum_{i=1}^n P_i\qquad\text{and}\qquad f(P_1)=\cdots=f(P_n), 
\end{equation}
    if  $r=n^a q$ for some $a\in\N_0, q\in\N$ where $n$ and $q$ are coprime, and 
    \begin{equation}\label{eq:kineq}
        k\leq\frac{2(n^{a+1}-1)}{n-1}.
    \end{equation}
\end{theorem}
\noindent\emph{Remark:} Let us specify these for two cases:
\begin{enumerate}[(i)]
    \item $r=1$ and $n$ prime leads to $k\leq 2$.
    \item $n=r$ prime gives $k\leq 2(n+1)$
\end{enumerate}
\begin{proof}
    Using Lemma \ref{lem:balancedequi} this follows from the results in \cite{BasuKunduFlag,Nath2024} where the Fadell-Husseini index $\ker\left(\pi_{X}^*:H^*(EG/G;\F_n)\rightarrow H^*(X_G;\F_n)\right)$ of $X=\cF_r(\C^{rn})$ was computed for the cyclic subgroup $G=\Z_n$ for prime $n$ and singular cohomology with coefficients in the field $\F_n$. 
    
    For $n$ an odd prime, the underlying cohomology ring is (computed e.g. in \cite{Dieck+1987}, p.187):
    \begin{equation}
    H^*(E\Z_n/\Z_n;\F_n)\simeq \F_n[u,v]/(u^2),\qquad \deg u =1,\quad \deg v=2.
    \end{equation} Hence, $u^2=0$, which implies that $H^{2l}$ is generated by $v^l$ and $H^{2l+1}$ by $v^l u$. For this case, it was shown in \cite{BasuKunduFlag} that $\ker(\pi_X^*)$ is the ideal generated by the two elements $(u v^{N-1},v^N)$ where $N:=n^{a+1}$.

    Setting $Y:=S(W_n(\R^k))$, this has to be compared with $\ker(\pi_Y^*)$, which turns out to be the ideal generated by $v^{k(n-1)/2}$ \cite{Volovikov}. An inclusion $\ker(\pi_Y^*)\subseteq\ker(\pi_X^*)$ would thus require $v^{k(n-1)/2}\in (v^N)$, which is equivalent to $k(n-1)/2 \geq N$, and fails to hold if Eq.(\ref{eq:kineq}) is true. 

    For $n=2$, the argument is similar: the underlying cohomology ring is $\F_2[t]$, with $\deg t=1$ (\cite{Dieck+1987}, p.187), so that $H^l\simeq\F_2\cdot t^l$. In \cite{Nath2024} it was shown that $\ker(\pi_X^*)$ is the ideal generated by $t^M$ with $M:=2^{a+2}-1$. In this case, $\ker(\pi_Y^*)$ is the ideal generated by $t^k$ (cf. \cite{BLAGOJEVIC20111326}, Thm.3.13).
    So an inclusion $\ker(\pi_Y^*)\subseteq\ker(\pi_X^*)$ would require $M\leq k$, which is again violated by Eq.(\ref{eq:kineq}).
\end{proof}
The remaining part of this section will collect some consequences of Thm.\ref{thm:equivP}:

\begin{corollary}[Stronger RA decomposition]\label{cor:strongerRA}
    Let $T$ be any cptp endomorphism on $\C^{2\times 2}$ and $g$ any continuous map from the set of those endomorphisms into $\R^2$. Then there exist cptp maps $T_1, T_2$ that are generalized extreme points (i.e. have Choi matrices of rank at most two) s.t.
    \begin{equation}
        T=\frac12\big(T_1+T_2\big)\quad\text{and}\quad g(T_1)=g(T_2).
    \end{equation}
\end{corollary}
\begin{proof}
    We apply Thm.\ref{thm:equivP} with $r=n=2$ to the Choi matrix of $T$. For the two positive components in the resulting decomposition to correspond to proper Choi matrices (up to a normalizing factor of $2$), they need to have equal values of ${\rm tr}_2[\cdot]$, which is then fixed by linearity. Within the Hermitian matrices, these are $2\times 2=4$ real constraints. Since  Thm.\ref{thm:equivP} allows for $k=6$ real constraints, we can use the remaining two to balance the components of $g$. The required continuous extension of $g$ is always possible due to the Tietze extension theorem and the closedness of the set of Choi matrices.  
\end{proof}
\begin{corollary}[Weak RA for $d=3$]\label{cor:weak3RA}
    The weak RA conjecture holds for $d_1=d_2=3$. That is, for any cptp endomorphism $T$ on $\C^{3\times 3}$ there are generalized cptp extreme points $T_1,T_2,T_3$ and weights $\lambda\in\R_+^3$ with $\sum_{i=1}^3\lambda_i=1$ s.t.
    \begin{equation}
        T=\sum_{i=1}^3 \lambda_i T_i.
    \end{equation}
\end{corollary}
\begin{proof}
    Again, this follows from Thm.\ref{thm:equivP} applied to the Choi matrix of $T$, now with $r=n=3$. The constraint function $f$ is chosen linear and such that it picks out the $8$ traceless components of the partial trace ${\rm tr}_2[\cdot]$. In this way, the resulting matrices in the decomposition become proper Choi matrices up to a positive scalar multiple, which is taken care of by $\lambda$.
\end{proof}

\begin{corollary}[Decomposition with similar marginals]\label{cor:marg}
    Let $n$ be prime,  $\rho\in{\rm Herm}(\C^n\otimes\C^n)$ be any density operator and $g:{\rm Herm}(\C^n\otimes\C^n)\rightarrow\R^2$ continuous. Then there are density operators $(\rho_i)_{i=1}^n$ of rank at most $n$ s.t. $g(\rho_i)=g(\rho_j)$ for all $i,j$, 
    \begin{equation}  \label{eq:eqispecmarg}      \rho=\frac1n\sum_{i=1}^n\rho_i\quad\text{and}\quad \spec\big({\rm tr}_\alpha[\rho_i]\big)= \spec\big({\rm tr}_\alpha[\rho_j]\big)\quad\forall i,j\ \forall \alpha\in\{1,2\},
    \end{equation}
    where $\spec(\cdot)$ denotes the spectrum as a multiset. 
\end{corollary}
\begin{proof}
    We employ Thm.\ref{thm:equivP} with $r=n$ and use that $A,B\in{\rm Herm}(\C^n)$ have the same spectrum iff $\tr{A^l}=\tr{B^l}$ for all $l=1,\ldots,n$. Choosing $f(Z):=(\tr{Z_1^l},\tr{Z_2^l})_{l=1}^n \in\R^{2n}$ with $Z_\alpha:=n\cdot {\rm tr}_\alpha[Z]$ then leads to the decomposition in Eq.(\ref{eq:eqispecmarg}) with $P_i=\rho_i/n$. Moreover, as Thm.\ref{thm:equivP} allows $k=2n+2$ real constraints, we can append the two components of $x\mapsto g(n x)$ to $f$ and obtain the sought decomposition.
\end{proof}

\begin{corollary}[Extended Schur-Horn decomposition]\label{cor:SHext}
    Let  $\rho\in{\rm Herm}(\cH)$ be any density operator, $N\in\N$  s.t. $\rank{\rho}\leq N$, and $g$ any real-valued continuous function on ${\rm Herm}(\cH)$. Assume that either 
    \begin{enumerate}
        \item $N$ is prime, or
        \item $g$ is linear.
    \end{enumerate}Then there exist $\varphi_i\in\cH$ with $\|\varphi_i\|=1$ s.t.
    \begin{equation}
        \rho=\frac1N\sum_{i=1}^N |\varphi_i\rangle\langle\varphi_i| \quad\text{and}\quad g(|\varphi_i\rangle\langle\varphi_i|)=g(|\varphi_j\rangle\langle\varphi_j|)\quad\forall i,j.
    \end{equation}
\end{corollary}
\begin{proof} \emph{(1)} Suppose $N$ is prime. Then the statement follows from Thm.\ref{thm:equivP}  with $n=N, r=1$ by using $f(\cdot):=\big(\tr{\cdot},\tilde{g}(\cdot)\big)$, where $\tilde{g}:\sigma\mapsto g(n\sigma)$. 

\emph{(2)} If $N$ is composite, choose a prime divisor $p_1\; |\; N$ and write $N=p_1\cdot m$. We apply Thm.\ref{thm:equivP} to $P=\rho$ with $n=p_1$, $r=m$ and $f(\cdot):=\big(\tr{\cdot},g(\cdot)\big)$. This is justified as $k=2$ always fulfills Eq.(\ref{eq:kineq}). Then Thm.\ref{thm:equivP} leads to density matrices $\rho_i$ of $\rank{\rho_i}\leq m$ s.t. $\rho=\tfrac{1}{p_1}\sum_{i=1}^{p_1} \rho_i$ and, by linearity, $g(\rho_i)=g(\rho)$. Now we apply this procedure recursively to each $\rho_i$, i.e., we choose a prime divisor $p_2$ of $m$, apply Thm.\ref{thm:equivP} to $P=\rho_i$ with $n=p_2$ etc., until we end up with the last prime divisor for which \emph{(1)} completes the proof.
\end{proof}

While we do not know to what extent the restriction to prime $n$  in Cor.\ref{cor:marg} or the nonlinear case of Cor.\ref{cor:SHext} is necessary, the following implication is an example where the restriction to prime $n$ can be lifted:
\begin{proposition}[Balanced ONB]\label{prop:balONB} Let $\cP_1(\cH):=\{|\psi\rangle\langle\psi|\;|\; \psi\in\cH, \|\psi\|=1\}$ be the set of positive rank-one projectors on a Hilbert space $\cH\simeq\C^n$. For any continuous function $g:\cP_1(\cH)\rightarrow\R$,  there is an orthonormal basis $\{e_i\}_{i=1}^n$ of $\cH$ s.t. $g(|e_i\rangle\langle e_i|)=g(|e_j\rangle\langle e_j|)$ for all $i,j$.
    \end{proposition}
\begin{proof}
    We will first show that for prime $n$, this is indeed a consequence of Thm.\ref{thm:equivP}. In fact, applying Thm.\ref{thm:equivP} to $P=\1$, $f(\cdot):=(\tr{\cdot},g(\cdot))$, and  $r=1$ gives positive rank-one operators $P_i$ that satisfy $g(P_i)=g(P_j)$, $\1=\sum_{i=1}^n P_i$,  and $\tr{P_i}=1$. The latter two, however, imply that the $P_i$'s must be mutually orthogonal, and thus correspond to an orthonormal basis. 

    Now we make this argument obsolete by replacing it with one that holds for any $n\in\N$. To this end, fix any orthonormal basis $\{v_i\}_{i=1}^n$ of $\cH$ and define $\varphi(t):=\sum_{i=1}^n t_i v_i$ for any $t\in S^{n-1}\subseteq\R^n$, and $h:S^{n-1}\rightarrow\R$, $h(t):=g(|\varphi(t)\rangle\langle\varphi(t)|)$. By the Yamabe–Yujob{\^o} theorem \cite{YamabeYujobo}, there exists an orthonormal basis $\{\nu_i\}_{i=1}^n$ of $\R^n$, s.t. $h(\nu_i)=h(\nu_j)$. However, this also gives a complex orthonormal basis of $\cH$. 
\end{proof}

We end this section with a remark on the SIC-POVM and MUB existence problems. Setting the issue of  prime $n$ aside, note that Cor.\ref{cor:SHext} has a structure that is almost identical to the one in the SIC-POVM Proposition \ref{prop:SIC} (after normalization of $P$). However, one crucial difference, which is the reason why the equivariant topology tool cannot be applied out-of-the-box, is that it does not give control over the constant, which we denoted by $c$ in the template problem. Lem. \ref{lem:balancedequi} and Thm.\ref{thm:equivP} only guarantee that there is \emph{some} $c$, while the SIC-POVM and MUB problems require a specific one. This is not an issue for the RA-conjecture since in this case linearity of $f$ determines $c$.

\section{Convex decompositions of quantum channels}\label{sec:convex}

In this section, we will investigate convex decompositions of quantum channels from a more general perspective. Our main result will be a condition for the structural robustness of such a decomposition (Thm.\ref{thm:robustness}), which then has consequences for the RA conjecture (Thm.\ref{thm:RAstructure}). We will also state an upper bound on the Carathéodory number (Thm.\ref{thm:Cnr}) and analyze the (non-)existence of barycentric decompositions (Thm.\ref{thm:bary}).\vspace{5pt}

We begin by fixing notation and setting up the differential-geometric framework and tools for the main argument.
\subsection{Channels and Kraus operator manifolds} We denote by $\cT$ the space of completely positive maps $T:\C^{d_1\times d_1}\rightarrow\C^{d_2\times d_2}$ that satisfy $T^*(\1)=I$ for a chosen positive definite $I$. We set $I=\1$ if we want all $T$'s to be trace-preserving, but the following arguments only require a fixed $I>0$. With the map $g:V_r:=(\C^{d_2\times d_1})^r\rightarrow {\rm Herm}(\C^{d_1})$, $g(K):=\sum_{i=1}^r K_i^* K_i$ we can parametrize elements of $\cT$ that can be represented using $r$ Kraus operators $(K_i)_{i=1}^r\in V_r$ by elements of the preimage $g^{-1}(I)=:\cV_r$. Since $I$ is a regular value of $g$, the set $\cV_r$ is a smooth (Stiefel) manifold (cf.\cite{ItenColbeckManifolds}) whose tangent space at $K\in\cV_r$ is given by the kernel of the differential of $g$, i.e., $T_K\cV_r=\ker(\d{g}_K)$. Note that $\d{g}_K:V_r\rightarrow{\rm Herm}(\C^{d_1})$, $\d{g}_K(\dot{K})=\sum_{i=1}^r K_i^*\dot{K_i}+\dot{K_i^*}K_i$. It will be useful to view $V_r$ as a real vector space with inner product $\langle X,Y\rangle_\R:=\sum_{i=1}^r\langle X_i,Y_i\rangle_\R$, with $\langle X_i,Y_i\rangle_\R:=\Real\big(\tr{X_i^* Y_i}\big)$. Similarly, we equip ${\rm Herm}(\C^{d_1})$ with the Hilbert-Schmidt inner product. With respect to these inner products, we can derive the adjoint map \begin{equation}
    (\d{g}_K)^*:{\rm Herm}(\C^{d_1})\rightarrow V_r,\quad H \mapsto 2 (K_i H)_{i=1}^r .\label{eq:dgKadjoint}
\end{equation}
\subsection*{Choi map} Next, we will have a differential-geometric look at the transition from Kraus operators to Choi matrices. To this end, define an unnormalized maximally entangled state $|\Omega\rangle:=\sum_{i=1}^{d_1} |i\rangle\otimes|i\rangle$ and denote by $|A\rangle:=(\1\otimes A)|\Omega\rangle$ the vectorized version of any $A\in\C^{d_2\times d_1}$. This leads to $\langle A | B\rangle=\tr{A^* B}$ and an expression of the form ${\rm tr}_2[|A\rangle\langle B|]=(B^* A)^T$ for the partial trace over the second tensor factor. The \emph{Choi map}, which maps tuples of Kraus operators to Choi matrices, can then be defined as
\begin{equation}
    C:\cV_r\rightarrow\cA,\quad K\mapsto\sum_{i=1}^r|K_i\rangle\langle K_i|,
\end{equation} where $\cA:=\{H\in{\rm Herm}(\C^{d_1}\otimes\C^{d_2})|{\rm tr}_2[H]=I^T\}$ is the affine space of admissible Choi matrices. The tangent space of $\cA$ at any point $\gamma$ is thus $T_\gamma\cA=\ker({\rm tr}_2)$ and the differential of the Choi map at a point $K\in\cV_r$ becomes
\begin{equation}
    \d{C}_K:\underbrace{T_K\cV_r}_{\ker (\d{g}_K)}\hspace*{-7pt}\longrightarrow\underbrace{T_{C(K)}\cA}_{\ker({\rm tr}_2)},\quad \dot{K}\mapsto\sum_{i=1}^r |K_i\rangle\langle \dot{K_i}|+|\dot{K_i} \rangle\langle K_i|. 
\end{equation}
For any $Y\in{\rm Herm}(\C^{d_1}\otimes\C^{d_2})$ define an endomorphism $\cY$ on $\C^{d_2\times d_1}$ via $|\cY(A)\rangle:=Y (\1\otimes A)|\Omega\rangle$. Self-adjointness translates from $Y$ to $\cY$ in the sense that $\tr{\cY(A)^* B}=\tr{A^*\cY(B)}$ holds for all $A,B\in \C^{d_2\times d_1}$. This enables a useful Lagrange multiplier characterization of the orthogonal complement of the range of $\d{C}_K$:
\begin{lemma}\label{lem:cokernel}
    For any $K\in\cV_r$ and $Y\in {\rm Herm}(\C^{d_1}\otimes\C^{d_2})$ the following are equivalent:
    \begin{enumerate}[(i)]
        \item For all $\dot{K}\in T_K\cV_r\  :$  $ \  \langle Y,\d{C}_K(\dot{K})\rangle_\R=0, $
        \item There is a $\Lambda\in{\rm Herm}(\C^{d_1})$ s.t. $\cY(K_i)=K_i\Lambda$ holds for all $i\in\{1,\ldots,r\}$.
    \end{enumerate}
\end{lemma}
\begin{proof}
    $(i)\Rightarrow (ii):$ First note that we can rewrite the expression by using that 
\begin{equation}
    \langle Y,\d{C}_K(\dot{K})\rangle_\R
    =
2\sum_{i=1}^r\langle\cY(K_i),\dot{K}_i\rangle_\R,
    \qquad
    \forall \dot{K}\in T_K\cV_r=\ker(\d g_K).
    \label{eq:Lem7}
\end{equation}
Set
$ Z:=\bigl(\cY(K_i)\bigr)_{i=1}^r .$
By the assumption and Eq.~\eqref{eq:Lem7}, we have
$$
    \langle Z,\dot K\rangle_\R=0
    \qquad
    \forall \dot K\in\ker(\d g_K),
$$
Hence
$ 
    Z\in\ker(\d g_K)^\perp.
$
However,
$$
    \ker(\d g_K)^\perp
    =
    {\rm im}\big((\d g_K)^*\big)
    =
    \{2(K_iH)_{i=1}^r\mid H\in{\rm Herm}(\C^{d_1})\}
$$
by Eq.~\eqref{eq:dgKadjoint}. Therefore there exists
$H\in{\rm Herm}(\C^{d_1})$ such that for all $i$: 
$\cY(K_i)=2K_iH$.
Setting $\Lambda:=2H$ then gives 
$\cY(K_i)=K_i\Lambda$,  which proves $(ii)$.
    
    $(ii)\Rightarrow (i):$ This follows from Eq.(\ref{eq:Lem7}) together with the fact that
     $$
        2\sum_{i=1}^r\langle\cY(K_i),\dot{K_i}\rangle_\R =2\sum_{i=1}^r \langle K_i\Lambda,\dot{K_i}\rangle_\R= \tr{\Lambda \d{g}_K(\dot{K})},
     $$ which vanishes for all $\dot{K}\in\ker(\d{g}_K)$.
\end{proof}
\subsection*{Edge relation} The last preparatory piece is a relation between  pairs of elements of $\cT$ for which the following Lemma provides equivalent characterizations:
\begin{lemma}[Edge relation]\label{lem:edge} Let $K\in\cV_r,K'\in\cV_{r'}$ be two tuples of Kraus operators corresponding to maps $T,T'\in\cT$, respectively. Then the following are equivalent and depend solely on the support spaces of the two corresponding Choi matrices:
\begin{enumerate}[(i)]
    \item ${\rm span}_\C\{K_i^*K'_j\}_{i,j}=\C^{d_1\times d_1}$.
    \item $T^*\circ T'$ has maximal Kraus rank  (i.e., $d_1^2$ lin. indep. Kraus operators or, equivalently, a full rank Choi matrix).
\end{enumerate}
\end{lemma}
\begin{proof}
    $(i)\Leftrightarrow (ii)$ is evident from the fact that a possible set of Kraus operators of $T^*\circ T'$ is precisely $\{K_i^*K'_j\}_{i,j}$. 

    That $(i)$ depends only on the support spaces of the Choi matrices becomes clear by noting that the support space of the Choi matrix $C(K)$, for instance, is given by ${\rm span}_\C\{|K_i\rangle \}_i$, which is linearly isomorphic to ${\rm span}_\C\{K_i\}_i$. Hence, the support spaces of the two Choi matrices determine ${\rm span}_\C\{K_i^*K'_j\}_{i,j}$. 
\end{proof}
\subsection{Structurally robust decompositions}\label{sec:robust}
\begin{theorem}[Structurally robust decompositions]\label{thm:robustness}
    For any $m$-tuple of maps $(T^{(k)}\in\cT)_{k=1}^m$  consider a graph $G$ that assigns a vertex to each of the $m$ maps and connects two vertices by an edge if the  relation in Lemma \ref{lem:edge} holds for the corresponding pair of maps. 
    
    If $G$ is connected, then for any convex combination $T=\sum_{k=1}^m\lambda_k T^{(k)}$ with full Kraus rank  and weights $\lambda_k\in(0,1)$ there exists an open neighborhood $U$ of $T$ in $\cT$ s.t. for every $\tilde{T}\in U$ there exist $(\tilde{T}^{(k)}\in\cT)_{k=1}^m$ where each $\tilde{T}^{(k)}$ has the same Kraus rank  as $T^{(k)}$ and $\tilde{T}=\sum_{k=1}^m\lambda_k \tilde{T}^{(k)}$.
\end{theorem}
\begin{proof}
    Denote by $r_k$ the Kraus rank  of $T^{(k)}$ and let $\mathsf{K}:=(K^{(1)},\ldots,K^{(m)})\in\cV_{r_1}\times\ldots\times\cV_{r_m}=:\cV$ be  tuples of Kraus operators so that $K^{(k)}$ parametrizes $T^{(k)}$. The goal is to show that the map
    \begin{equation}
        F:\cV\rightarrow\cA,\quad F(\mathsf{K}):=\sum_{k=1}^m \lambda_k C(K^{(k)})\label{eq:F}
    \end{equation}
    is locally open around $\mathsf{K}$. By the submersion theorem it suffices to show that the differential $\d{F}_\mathsf{K}$ is surjective. Due to the linear relation between $F$ and $C$ this is equivalent to showing that
    \begin{equation}
        \forall k\in\{1,\ldots,m\}\ \forall \dot{K}\in T_{K^{(k)}}\cV_{r_k}:\qquad \langle Y,\d{C}_{K^{(k)}}(\dot{K})\rangle=0\label{eq:q2fe33}
    \end{equation}
    implies $Y=0$. From Lemma \ref{lem:cokernel} we know that Eq.(\ref{eq:q2fe33}) is equivalent to the existence of a set of Hermitian $\Lambda^{(k)}$ s.t. $\cY(K_i^{(k)})=K_i^{(k)}\Lambda^{(k)}$ holds for all $i,k$. 

    Suppose, as we will verify later, that $\Lambda^{(k)}=\Lambda$ is independent of $k$. Then $\cY(A)=A\Lambda$ holds for all $A\in{\rm span}_\C\{K_i^{(k)}\}_{i,k}$. As $T$ is assumed to have full Kraus rank , the span of all Kraus operators is the entire matrix space. So $A$ ranges all over $\C^{d_2\times d_1}$. By the definition of $\cY$, the equation can be rewritten as $Y(\1\otimes A)|\Omega\rangle=(\1\otimes A\Lambda)|\Omega\rangle=(\Lambda^T\otimes A)|\Omega\rangle$, and since this has to hold for all $A$, it simplifies to $Y=\Lambda^T\otimes\1$. However, $Y\in\ker({\rm tr}_2)$ which means that ${\rm tr}_2[Y]=\Lambda^T d_2=0$, which in turn implies $\Lambda=0$ and finally $Y=0$, as desired.

    It remains to show that $\Lambda^{(k)}$ is indeed independent of $k$. To this end, we exploit self-adjointness of $\cY$ and write
    \begin{eqnarray}
        0&=& \tr{\cY\big(K_i^{(k)}\big)^* K_j^{(l)}}-\tr{\cY\big(K_j^{(l)}\big) K_i^{(k)*}}\nonumber\\ \nonumber
        &=& \tr{\big(K_i^{(k)}\Lambda^{(k)}\big)^* K_j^{(l)}}-\tr{\big(K_j^{(l)}\Lambda^{(l)}\big) K_i^{(k)*}}\\
         &=& \tr{\big(\Lambda^{(k)}-\Lambda^{(l)}\big)K_i^{(k)*}K_j^{(l)}}.\label{eq:arrayedge}
    \end{eqnarray}
    If the edge relation from Lemma \ref{lem:edge} holds for the pair $T^{(k)},T^{(l)}$, then Eq.(\ref{eq:arrayedge}) implies that $\Lambda^{(k)}=\Lambda^{(l)}$. In this way, the assumed connectedness of the graph $G$ finally guarantees that all $\Lambda^{(k)}$ are equal.
\end{proof}
\subsection*{Example:} In order to show that the above assumptions are nonvacuous, we will construct an example for $r_k=d_1, m=d_2$. To this end, let $S\in\C^{d_2\times d_2}$ implement a cyclic shift of the standard basis, i.e., $S|k\rangle=|k+1\rangle$, cyclically, and let $\psi_0\in\C^{d_2}$ be any unit vector s.t. $\langle\psi_0,S\psi_0\rangle\neq 0$. For $k=0,\ldots,d_2-1$ define $\psi_k:=S^k\psi_0$ and \begin{equation}
    K_j^{(k)}:=|\psi_k\rangle\langle j|\sqrt{I},\label{eq.Kex1}
\end{equation}
 where $\{|j\rangle\}$ is an orthonormal basis of $\C^{d_1}$. By construction, all these $K^{(k)}$ parametrize an element $T^{(k)}\in\cT$, and since for a fixed $k$, the $d_1$ operators in Eq.(\ref{eq.Kex1}) are linearly independent, $T^{(k)}$ has Kraus rank  $d_1$. In order to verify the connectedness of the associated graph $G$, we exploit that $\langle\psi_k,\psi_{k+1}\rangle=\langle\psi_0,S^{-k}S^{k+1}\psi_{0}\rangle=\langle\psi_0,S\psi_{0}\rangle=:c$. So, by invertibility of $I$, for every $k$ $$K_j^{(k)*}K_i^{(k+1)}=c\sqrt{I}|j\rangle\langle i|\sqrt{I}$$ spans all of $\C^{d_1\times d_1}$ when varying $i$ and $j$. Therefore, $G$ contains a Hamiltonian cycle and is thus connected. Finally, we need to ensure that $T$ has full Kraus-rank. To this end, note that $\dim{\rm span}\{K_j^{(k)}\}=d_1\cdot\dim{\rm span}\{S^k\psi_0\}$. Hence, $T$ has full Kraus rank if we choose $\psi_0$ s.t. $\{S^k\psi_0\}$ span the whole space, which is a generic condition.  \vspace*{5pt}

An immediate consequence of this result is that the strong RA conjecture holds in any dimension on a set of nonzero measure:

\begin{theorem}[Structure of the strong RA set]\label{thm:RAstructure} For any $d_1,d_2\in\N$,
    let $\tilde{\cT}$ be the set of cptp maps $\C^{d_1\times d_1}\rightarrow\C^{d_2\times d_2}$ that satisfy the strong RA conjecture. As a subset of the affine space of all trace-- and hermiticity-preserving linear maps,  $\tilde{\cT}$ is closed, connected, semialgebraic, and has nonempty interior. In particular, it has nonzero Lebesgue-measure.
\end{theorem}
\begin{proof}
    That $\tilde{\cT}$ has nonempty interior follows from the fact that the foregoing example satisfies the strong RA conjecture (when all weights are chosen equal) and that it has a neighborhood with the same property by Thm.\ref{thm:robustness}.

    That $\tilde{\cT}$ is semialgebraic is a simple consequence of the Tarski-Seidenberg theorem since $\tilde{\cT}$ can be expressed in first-order logic over $\R$ using only polynomial (in-)equalities. 

    Closedness of $\tilde{\cT}$ follows from compactness of the set $\cT_{d_1}\subseteq\cT$ of cptp maps with Kraus rank at most $d_1$. In order to see this, consider a sequence $\big(T^{(n)}\in\tilde{\cT}\big)_{n\in\N}$ converging to $\lim_{n\rightarrow\infty} T^{(n)}=T\in\cT$. We want to show that $T\in\tilde{\cT}$. For each $T^{(n)}$ let $\big(T^{(n)}_1,\ldots, T^{(n)}_{d_2}\big)\in\cT_{d_1}^{d_2}$ be the elements in the corresponding equal weight RA decomposition. Since $\cT_{d_1}^{d_2}$ is compact, there is a subsequence $\big(T^{(n_m)}_i\big)_{i=1}^{d_2}$ converging in $\cT_{d_1}^{d_2}$ for $m\rightarrow\infty$. Then $\lim_{m\rightarrow\infty} T^{(n_m)}_i =T_i$ converges in $\cT_{d_1}$ so that 
    \begin{equation}
        T=\lim_{m\rightarrow\infty}\frac{1}{d_2}\sum_{i=1}^{d_2} T^{(n_m)}_i = \frac{1}{d_2}\sum_{i=1}^{d_2} T_i.
    \end{equation}
    Finally, connectedness follows from connectedness of $\cT_{d_1}$. The latter, in turn, can be seen by noting that by Stinespring's theorem every $\phi\in\cT_{d_1}$ can be represented as $\phi(X)={\rm tr}_R[VXV^*]$ by an isometry $V:\C^{d_1}\rightarrow\C^{d_2}\otimes\cH_R$, where $\cH_R\simeq\C^{d_1}$. In this way, the connectedness of the set of isometries leads to connectedness of $\cT_{d_1}$. 
\end{proof}

\subsection{Carathéodory number and barycentric decompositions}

Before we prove the strong RA conjecture for specific classes of quantum channels  in the following section, we add some small results on the Carathéodory number and the existence of barycentric decompositions for quantum channels. 

As mentioned earlier, Carathéodory's theorem guarantees that every cptp map $T:\C^{d_1\times d_1}\rightarrow\C^{d_2\times d_2}$ can be convexly decomposed into at most $(d_2^2-1)d_1^2+1$ extreme points. This bound can easily  be improved as follows:
\begin{theorem}[Carathéodory number bound]\label{thm:Cnr} Every cptp map $T:\C^{d_1\times d_1}\rightarrow\C^{d_2\times d_2}$ can be convexly decomposed into at most $r$ extreme points, where $r\leq d_1 d_2$ is the Kraus rank of $T$.    
\end{theorem}
\emph{Remark:} For $d_1=d_2$ this bound can already be found in \cite{loewy2016facessetquantumchannels} and it can also be seen as a consequence of the facial chain length bound in \cite{Chainlength}.
\begin{proof}
    Let $\phi\in\cT$ be an extreme point that appears with nonzero weight in a convex decomposition of $T$, and let $P_\phi, P_T\in{\rm Herm}(\C^d_1\otimes\C^{d_2})$ be the  Choi matrices corresponding to $\phi$ and $T$, respectively. Then $\lambda_{max}:=\max\{\lambda\in\R|P_T-\lambda P_\phi\geq 0\}$ is in the interval $(0,1)$, and $P_\Psi:=(P_T-\lambda_{max} P_{\phi})/(1-\lambda_{max})$ is a valid Choi matrix, which characterizes a cptp map $\Psi$ of Kraus rank at most $r-1$. 
    Hence, we have $$T=\lambda_{max}\phi+(1-\lambda_{max})\Psi,$$ which is a convex decomposition of a Kraus rank $r$ channel into an extreme point and a Kraus rank $\leq r-1$ channel. Iterating this successively reduces the Kraus rank while adding another extreme point in every step.  
\end{proof}
Next we show that a finite barycentric convex decomposition of a quantum channel is always possible in terms of generalized extreme points, but not always in terms of extreme points:
\begin{theorem}[Barycentric decompositions]\label{thm:bary} 
     Let $\cT$ be the convex set of cptp maps $T:\C^{d_1\times d_1}\rightarrow\C^{d_2\times d_2}$, $\cE\subseteq\cT$ the set of extreme points and $\overline{\cE}$ its topological closure.
     \begin{enumerate}
         \item Every $T\in\cT$ admits a barycentric decomposition into elements of $\overline{\cE}$. That is, there are $T_1,\ldots,T_n\in\overline{\cE}$ for some $n\in\N$, s.t.
         \begin{equation}\label{eq:baryc}
             T=\frac1n\sum_{i=1}^n T_i.
         \end{equation}
         Here, one can always choose $n\leq 2^{\max\{0,r-d_1\}}$, where $r\leq d_1 d_2$ is the Kraus rank of $T$. 
         \item If $d_1\geq 2, d_2\geq 3$, there are $T\in\cT$ that do not admit a barycentric decomposition of the form in Eq.(\ref{eq:baryc}) with $T_i\in\cE$ for any $n\in\N$.
     \end{enumerate}
\end{theorem}
\begin{proof}
    \emph{(1)} Let $P\in {\rm Herm}(\C^{d_1}\otimes\C^{d_2})$ be the Choi matrix of $T$ and assume that $r:=\rank{P}>d_1$ (otherwise, there is nothing to show). Define $F_P:=\{Q\geq 0| {\rm supp}[Q]\subseteq{\rm supp}(P),{\rm tr}_2[Q]={\rm tr}_2[P]\}$. This is a compact convex face within the set of all Choi matrices, with $P\in{\rm relint}[F_P]$ in its relative interior. The dimension of this face, i.e., the dimension of its affine hull, is $\dim_{\R}(F_P)\geq 3$. One way to see this is to note that $\dim_{\R}(F_P)=\dim_{\R}(\ker[\Gamma])$, where $\Gamma:{\rm Herm}({\rm supp}[P])\rightarrow {\rm Herm}(\C^{d_1}), \Gamma(X):={\rm tr}_2[X]$. Rank-nullity then gives $\dim_{\R}(\ker[\Gamma])\geq r^2-d_1^2\geq (d_1+1)^2-d_1^2=2d_1+1\geq 3$.
    Now we use the following elementary convex-geometric fact, which is a consequence of the intermediate value theorem: for any point $P$ in the relative interior of a convex compact set of dimension at least two, there are points $P_1,P_2$ on the relative boundary s.t. $P=(P_1+P_2)/2$. In our case, $P_1,P_2$ are Choi matrices of  $\rank{P_i}<r$. 

    Now we can iterate this argument. Applying it to each of $P_1,P_2$ reduces the rank further. Since the rank drops by at least one in every branch and we stop when no rank is larger than $d_1$, we need at most $r-d_1$ iterations, leading to at most $2^{r-d_1}$ terms in the decomposition. Note that once the rank drops to  $d_1$ or below in one branch, we can still decompose the corresponding term further in the trivial way $P=(P+P)/2$ in order to guarantee equal weights.

    \emph{(2)} The reason that in the case of $d_1\geq 2, d_2\geq 3$ a barycentric decomposition into extreme points does not always exist is that in this case $\cT$ has one-dimensional faces (which we will prove subsequently). Since every element in such a face has a unique decomposition into the two extreme points, every point in the relative interior whose decomposition into those extreme points involves irrational weights then does not admit a barycentric decomposition into extreme points. 

    Now, let us construct a one-dimensional face. Using orthonormal bases, define operators $K_1,\ldots,K_{d_1}:\C^{d_1}\rightarrow\C^{d_2}$ by $K_1:=|1\rangle\langle1|+|2\rangle\langle 2|$, $K_2:=|3\rangle\langle 1|+|1\rangle\langle 2|$, and, if $d_1>2$, then $K_i:=|1\rangle\langle i|$ for $i\in\{3,\ldots,d_1\}$. Consider the set $F:=\{P\in {\rm Herm}(\C^{d_1}\otimes\C^{d_2})|P\geq0,{\rm supp}(P)\subseteq\cK,{\rm tr}_2{P}=\1 \}$ of proper Choi matrices supported in $\cK:={\rm span}_{\C}\{|K_1\rangle,\ldots,|K_{d_1}\rangle\}$. This is clearly a face, since Choi matrices with support outside of $\cK$ can never appear in any convex decomposition of any element of $F$. Every element of $F$ is of the form $J=\sum_{i,j} H_{i,j}|K_i\rangle\langle K_j|$ for some Hermitian $H\in\C^{d_1\times d_1}$. $H$ is constrained by the positivity of $J$ and the trace-preserving constraint $\gamma(H)=\1$ where $\gamma$ is an endomorphism on ${\rm Herm}(\C^{d_1})$ given by $\gamma(H):=\sum_{i,j}H_{i,j}K_i^* K_j$. Consequently,  $\dim(F)\leq \dim(\ker(\gamma))$. In order to determine the kernel of $\gamma$ note that 
    \begin{eqnarray}
        K_1^*K_1=K_2^* K_2&=&|1\rangle\langle1|+|2\rangle\langle2|\nonumber \\
        K_1^*K_2=|1\rangle\langle2|&,& K_2^*K_1=|2\rangle\langle1|\\
        K_j^*K_k=|j\rangle\langle k|&& \text{for all }k\geq 3 \text{ and all }j.\nonumber
    \end{eqnarray}
Hence, there is only one linear relation and we obtain a one-dimensional kernel $\ker(\gamma)=\R\cdot {\rm diag}(1,-1,0,\ldots,0)$. That the dimension of $F$ is not reduced further by the positivity constraint is shown by the explicit one-dimensional family
$$\lambda |K_1\rangle\langle K_1|+(1-\lambda)|K_2\rangle\langle K_2|+\sum_{i\geq 3} |K_i\rangle\langle K_i|\ \in\  F,\quad\text{for all }\lambda\in[0,1].$$
    
\end{proof}
The case $d_1=d_2=2$, which is omitted in Thm.\ref{thm:bary} \emph{(2)}, is special: it follows from the results in \cite{loewy2016facessetquantumchannels} that in this case $\cT$ does not have one-dimensional faces. Consequently, the construction in the proof of Thm.\ref{thm:bary} \emph{(1)} still works with $T_i\in\cE$ so that every qubit channel does admit an equal-weight decomposition into extreme points.  

\section{Classes satisfying the strong RA conjecture}\label{sec:strongRAexp}
It has been shown in \cite{Szarek} that every quantum channel with a two-dimensional output ($d_2=2$) satisfies the strong RA conjecture, i.e. it can be written as the average of two channels with Kraus rank at most $d_1$. This is essentially a consequence of the fact that every contraction can be represented as an average of two unitaries. 

The following shows that any channel with a two-dimensional input also satisfies the strong RA conjecture. Note that this is a substantially different result as it concerns an equal-weight mixture of a large number of low-Kraus rank channels (rather than a small number with high Kraus-rank) — also, the proof is quite different:

\begin{theorem}[Strong RA for $d_1=2$]\label{thm:strong2RA}
    The strong RA conjecture holds for any cptp map $T:\C^{d_1\times d_1}\rightarrow \C^{d_2\times d_2}$ where $d_1=2$. That is, in this case, there are cptp maps $T_1,\ldots,T_{d_2}$, each of Kraus rank at most two, s.t. \begin{equation}
        T=\frac{1}{d_2}\sum_{i=1}^{d_2} T_i.\label{eq:sRATis}
    \end{equation}
\end{theorem}
\begin{proof}
    Let $P\in{\rm Herm}(\C^{d_1}\otimes\C^{d_2})$ be the Choi matrix of $T$, which implies ${\rm tr}_2[P]=\1$, and $P\geq 0$. W.l.o.g. we can assume that $P$ is, in fact, positive definite as the singular case will eventually be covered by compactness and continuity.

    On the level of Choi matrices, the decomposition in Eq.(\ref{eq:sRATis}) corresponds to a decomposition $P=\sum_{i=1}^{d_2} P_i$ where $P_i\geq 0$ and ${\rm tr}_2[P_i]=\1/d_2$ for all $i$. If we neglect the partial trace constraint, we can parametrize a family of decompositions (in fact, \emph{all} if $P>0$) by a single unitary $U\in U(d_1 d_2)$ via \begin{equation}
        P_i=\sqrt{P}U^*\underbrace{\big(\1_{d_1}\otimes|i\rangle\langle i|\big)}_{=:E_i} U \sqrt{P},\label{eq:PiU}
    \end{equation}
where $\{|i\rangle\in\C^{d_2}\}$ is an orthonormal basis. The aim is to show that there is a $U$ s.t. all $a_i(U):={\rm tr}_2[P_i]$ have the same value, which is then $a_i(U)=\1/d_2$. Hence, we seek a zero of the functional 
\begin{equation}
    f(U):=\sum_{i=1}^{d_2}\|a_i(U)-\1/d_2\|_2^2.
\end{equation}
Consider the variation $\varphi(t):=f\big(U(t)\big)$, $t\in\R$  with $U(t):=\exp{[tK]}U$, where we choose a skew-Hermitian generator $K:=Z\otimes|p\rangle\langle q|-Z^*\otimes |q\rangle\langle p|$, using the tensor product structure and basis that appears in the $E_i$'s, and $Z\in\C^{d_1\times d_1}$.
 A straightforward, albeit lengthy, calculation leads to the derivatives  
\begin{eqnarray}
    \varphi'(0)&=&4 \Real \sum_{i=1}^{d_2} \tr{\big(a_p(U)-a_q(U)\big)(Y_{pi})^* Z\; Y_{qi}}\label{eq:varhi10}, \\ \label{eq:varhi20} 
    \varphi''(0)&=& 16\Big\|\Real\sum_{i=1}^{d_2} (Y_{pi})^* Z\; Y_{qi}\Big\|_2^2 + \\  \nonumber && +4\tr{\big(a_p(U)-a_q(U)\big)\sum_{i=1}^{d_2}\big((Y_{qi})^* Z^*Z Y_{qi} - (Y_{pi})^* ZZ^* Y_{pi}\big)}\nonumber,
\end{eqnarray}
where we use the block-matrix decomposition $Y=\sum_{i,j=1}^{d_2} Y_{ij}\otimes |i\rangle \langle j|:=U\sqrt{P}$.

If we restrict to $Z\in SU(d_1)$, then using  $a_q(U)-a_p(U)=\sum_i (Y_{qi})^*Y_{qi}-(Y_{pi})^* Y_{pi}$ and setting $H(Z):=\Real\sum_i (Y_{pi})^* Z Y_{qi}$,   Eq.(\ref{eq:varhi20}) simplifies to
\begin{equation}\label{eq:phi20su}
    \varphi''(0)= 16 \big\|H(Z)\big\|_2^2-4\big\|a_p(U)-a_q(U)\big\|_2^2.
\end{equation}

Suppose $U$ is a stationary point of $f$. Then $\varphi'(0)=0$ in Eq.(\ref{eq:varhi10}) for all $Z\in\C^{d_1\times d_1}$, which implies that in ${\rm Herm}(\C^{d_1})$ the range of $H$ is in the orthogonal complement of $a_p(U)-a_q(U)$, i.e., $\tr{\big(a_p(U)-a_q(U)\big)H(Z)}=0$. Thus, for $d_1=2$, the range of $H$ lies within a vector space isomorphic to $\R^3$, while the domain of $H$ is homeomorphic to $S^3\simeq SU(2)$ (as for instance the first column vector of $Z$ becomes a unit vector in $\R^4$, when separating real and imaginary parts). Moreover, $H$ is odd in the sense that $H(-Z)=-H(Z)$. Hence we get an odd map $S^3\rightarrow\R^3$, which, by the Borsuk-Ulam theorem, must have a zero. That is, there is a $Z\in SU(2)$ s.t. $H(Z)=0$. Consequently, Eq.(\ref{eq:phi20su}) yields $\varphi''(0)<0$ unless $a_p(U)=a_q(U)$. In other words,
every non-balanced stationary point has a negative second variation along some admissible curve. Therefore, every local minimum is balanced. Since $U(d_1 d_2)$ is compact, $f$ attains a global minimum, and that global minimum must be balanced.

\end{proof}

Finally, we show that if either the input or the output of a quantum channel is classical, meaning that it consists only of diagonal density matrices, then the strong RA conjecture holds. 
\begin{theorem}[Strong RA for cq- and qc-channels]\label{thm:cqqc}
    For any $d_1, d_2\in\N$, the strong RA conjecture holds for any cptp map $T:\C^{d_1\times d_1}\rightarrow \C^{d_2\times d_2}$ that is either classical-quantum (cq) or quantum-classical (qc). 
\end{theorem}
\begin{proof}
    \emph{cq-channels:} $T$ is a cq-channel if there is an orthonormal basis $|i\rangle\in\C^{d_1}$ together with density matrices $\rho_1,\ldots,\rho_{d_1}\in\C^{d_2\times d_2}$, s.t. $T(\rho)=\sum_{i=1}^{d_1}\langle i|\rho|i\rangle \rho_i$. Employing Prop.\ref{prop:SchurHorn} (Schur-Horn decomposition), we can write each of them as $\rho_i=\tfrac{1}{d_2}\sum_{j=1}^{d_2}|\varphi_{ij}\rangle\langle\varphi_{ij}|$. Hence, 
    \begin{equation}
        T(\rho)=\frac{1}{d_2}\sum_{j=1}^{d_2}T_j(\rho), \quad\text{where}\quad T_j(\rho):=\sum_{i=1}^{d_1} \langle i|\rho|i\rangle |\varphi_{ij}\rangle\langle\varphi_{ij}|,
    \end{equation}
    so that each $T_j$ is a cptp map with Kraus rank  at most $d_1$.\vspace*{5pt}

    \emph{qc-channels:} $T$ is a qc-channel if there is an orthonormal basis $|i\rangle\in\C^{d_2}$ and positive semidefinite operators $Q_1,\ldots,Q_{d_2}\in\C^{d_1\times d_1}$ with $\sum_i Q_i=\1$, s.t. $T(\rho)=\sum_{i=1}^{d_2} \tr{\rho Q_i} |i\rangle\langle i|$. Let $|\alpha\rangle\in\C^{d_1}$ denote the elements of an orthonormal basis and define $K_{m,\alpha}:=\sum_{j=1}^{d_2} |j\rangle\langle\alpha|\sqrt{Q_j} e^{2\pi i(m-1)j/d_2}$ for $m\in\{1,\ldots, d_2\}$. Then each $T_m(\rho):=\sum_{\alpha=1}^{d_1}K_{m,\alpha}\rho K_{m,\alpha}^*$ defines a cptp map whose Kraus rank  is at most $d_1$, and
    \begin{eqnarray}
        \frac{1}{d_2}\sum_{m=1}^{d_2} T_m(\rho) &=& \frac{1}{d_2}\sum_{j,k=1}^{d_2} \tr{\sqrt{Q_j}\rho\sqrt{Q_k}}|j\rangle\langle k|\underbrace{\sum_{m=1}^{d_2}\exp\Big[\frac{2\pi i}{d_2}(m-1)(j-k)\Big]}_{=d_2\cdot\delta_{j,k}}\nonumber \\ \nonumber &=& \sum_{j=1}^{d_2}\tr{Q_j \rho}\;|j\rangle\langle j|=T(\rho).
    \end{eqnarray}
\end{proof}

\section*{Acknowledgments}
MMW was supported by the Swedish Research Council under grant no. 2021-06594 while in residence at Institut Mittag-Leffler in Djursholm, Sweden in April 2026.

%
%

\appendix

\section{Symmetric informationally complete POVMs}\label{sec:SICPOVM}

 The term ``SIC - POVM'' was coined in \cite{Renes_2004} as a set of $d^2$ normalized vectors in $\C^d$ that satisfy an equiangular condition. 
 \begin{definition*}[SIC-POVM]
     A SIC-POVM is a set $\{\vec{\psi_i}\}_{i=1}^{d^2},\ \vec{\psi_i}\in\C^d$ that satisfies $\abs{\dotp{\psi_i, \psi_j}}^2=\begin{cases}
         \frac{1}{d+1}&i\neq j\\
         1&i=j
     \end{cases}$.
 \end{definition*}
 For a historical  overview the reader should refer to \cite{Fuchs_2017}. 

The existence of SIC-POVMs in all dimensions was first conjectured by Zauner \cite{Zauner,zauner1999} in his PhD thesis, although under a different name and again
conjectured under the name we know now in \cite{Renes_2004}. There are many equivalent problems, such as the minimization of a particular frame potential 
\cite{Kwapisz_2019} and showing that the entanglement breaking rank of a particular channel is $d^2$ \cite{Pandey_2020}. For certain dimensions, a construction of SIC-POVMs is implied by the existence of Stark units from a ray class field extension of a real quadratic field  \cite{kopp2018sicpovmsstarkconjectures}.

\vspace{6pt}
Prop.\ref{prop:SIC} is a reformulation of Thm.4 in \cite{Scott_2006}, where SIC-POVMs are characterized as minimal two-designs. 
As a reminder, a 2-design of size $N$ in $\C^d$ is given by $N$ unit vectors $\{\vec{\psi_i}\}_{i=1}^N$ in $\C^d$  such that the following holds:
\[\frac{1}{N}\suml{i=1}{N}\op{\psi_i}{\psi_i}\otimes\op{\psi_i}{\psi_i}=\frac{2}{d(d+1)}P_{sym}.\]
A 2-design is called \emph{minimal} if $N=d^2$.
Our reformulation of Thm.4 in \cite{Scott_2006} uses two small insights:
\begin{enumerate}[(\roman*)]
    \item The constraint $\tr{{\rm tr}_1[P_i]^2}=\tr{P_i}^2$, means that ${\rm tr}_1[P_i]^2$ has rank one so that $P_i$ is a product.
    \item Each $P_i$ has to be supported in the symmetric subspace. Combined with (i) this means that $P_i=\op{\psi_i}{\psi_i}\otimes\op{\psi_i}{\psi_i}$, which is the form used in Thm.4 in \cite{Scott_2006}. 
\end{enumerate}

In this way, Thm.4 in \cite{Scott_2006} becomes our proposition:
\begin{proposition*}[{\bf \ref{prop:SIC}} SIC-POVMs] Let $P_{sym}$ be the hermitian projector onto the symmetric subspace of $\C^d\otimes\C^d$. There exists a SIC-POVM in $\C^d$ iff there are positive rank-one operators $P_1,\ldots,P_{d^2}$ s.t. $\frac{2d}{d+1}P_{sym}=\sum_{i} P_i$ and for all $i$: $\tr{{\rm tr}_1[P_i]^2}=(\tr{P_i})^2$.  
\end{proposition*}

\section{Mutually unbiased bases (MUBs)}\label{sec:MUB}

Recall that two orthonormal bases $(u_i)_{i=1}^d$, $(v_j)_{j=1}^d$ of $\C^d$ are called \emph{mutually unbiased} if $|\langle u_i,v_j\rangle|^2=1/d$ for all $i,j$. It is known that at most $(d+1)$ MUBs can exist in dimension $d$ and that this bound is tight whenever $d$ is a prime or prime power---the three Pauli-bases in $\C^2$ being the best known example. In other dimensions, it is unknown whether this bound is saturated. For a review of the vast literature on the problem, including fourteen equivalent reformulations thereof, see \cite{McNulty2026mutuallyunbiased}.

In this subsection, we aim to prove Prop.\ref{prop:MUB}---yet another equivalent reformulation. 
To this end, we need some preliminary definitions and Lemmas. With $|\Omega\rangle:=\sum_{i=1}^d |ii\rangle\in\C^d\otimes\C^d$, $\omega:=|\Omega\rangle\langle\Omega|/d$ and $R_\1:=-w+\sum_{i=1}^d |ii\rangle\langle ii|$, we define a Hermitian rank-$(d-1)$ projector
\begin{equation}
    R_U:=(U\otimes \overline{U})R_\1 (U\otimes \overline{U})^*
\end{equation} for any unitary $U\in U(d)$. Note that $|\Omega\rangle$ is in the kernel of every $R_U$. If we think of an orthonormal basis as the column vectors of a unitary, the MUB property can be expressed in terms of the $R_U$'s as follows:
\begin{lemma}[Orthogonal projections vs. MUBs]\label{lem:OrthopMUBs}\ 
    \begin{enumerate}
        \item $U,V\in U(d)$ are mutually unbiased iff $R_U R_V=0$.
        \item A set of $(d+1)$ unitaries  $U_1,\ldots,U_{d+1}\in U(d)$ is MUB iff 
        \begin{equation}\label{eq:RUisum}
            \sum_{i=1}^{d+1} R_{U_i}=\1-\omega.
        \end{equation}
    \end{enumerate}
\end{lemma}
\begin{proof}
    \emph{(1)} Denoting by $u_i, v_j$ the column vectors of $U$ and $V$, respectively, the assertion follows from realizing  that the support of $R_U$ is spanned by $(u_i\otimes\overline{u}_i-\Omega/d)_{i=1}^d$ (similarly for $R_V$) and that $$\langle u_i\otimes\overline{u}_i-\Omega/d,v_j\otimes\overline{v}_j-\Omega/d\rangle = |\langle u_i,v_j\rangle|^2 - \frac1d .$$
    \emph{(2)} If the $U_i$'s are MUBs, then \emph{(1)} implies that their sum is a Hermitian projector of rank $(d+1)(d-1)=d^2-1$. The kernel must contain $|\Omega\rangle$ and since it is one-dimensional, Eq.(\ref{eq:RUisum}) follows.

    The converse follows from the fact that the l.h.s. of Eq.(\ref{eq:RUisum}) is a sum of $(d+1)$ Hermitian projectors, each of rank $(d-1)$, while the r.h.s. is a Hermitian projector of rank $(d^2-1)$. This requires the projectors in the sum to be  mutually orthogonal, which in turn implies the MUB property by \emph{(1)}.
\end{proof}
The next Lemma recovers the structure of $R_U$ from a constrained projection.
\begin{lemma}\label{lem:XRU}
    If $X\in{\rm Herm}(\C^d\otimes\C^d)$ is a Hermitian projector of rank $d$ that satisfies ${\rm tr}_1[X]={\rm tr}_2[X]=\1$, $X|\Omega\rangle=|\Omega\rangle$ and $X^{T_1}\geq 0$, then there is an orthonormal basis $(a_i)_{i=1}^d$ of $\C^d$ s.t.
    \begin{equation}
        X=\sum_{i=1}^d |a_i\otimes\overline{a}_i\rangle\langle a_i\otimes\overline{a}_i|.
    \end{equation}
    Consequently, there is a unitary $U\in U(d)$ s.t. $X-\omega=R_U$.
\end{lemma}
\begin{proof}
    Following \cite{PPT}, ${\rm rank}(X)=d$ together with $X^{T_1}\geq 0$ implies the existence of a separable decomposition of length $d$, i.e. 
    \begin{equation}
        X=\sum_{i=1}^d |a_i\otimes b_i\rangle\langle a_i\otimes b_i|.
    \end{equation}
    In order for $X$ to have rank $d$, the vectors $|v_i\rangle:=|a_i\otimes b_i\rangle$ have to be linearly independent. Defining $V|i\rangle:=|v_i\rangle$, we get $X= VV^*$ and $G:=V^*V\in GL(d)$. Then $X^2=X$ implies $G^3=G^2$ and therefore $G^2=G=\1$. So the $v_i$'s are orthonormal.  

    The partial trace constraint $\1={\rm tr}_1[X]=\sum_i \|a_i\|^2 |b_i\rangle\langle b_i|$ implies orthogonality of the $b_i$'s and similarly, $\1={\rm tr}_2[X]$ implies orthogonality of the $a_i$'s. As $1=\|v_i\|=\|a_i\|\cdot \|b_i\|$ we can choose $(a_i)_{i=1}^d$ and $(b_i)_{i=1}^d$ to be orthonormal bases.

    Finally, we utilize the constraint $X|\Omega\rangle=|\Omega\rangle$, which results in
    $$\sum_{i=1}^d \langle a_i,\overline{b}_i\rangle |a_i\otimes b_i\rangle =|\Omega\rangle.$$
    Acting from the left with $\1\otimes\langle b_j|$ on this equation leads to $|\overline{b}_j\rangle=\langle a_j,\overline{b}_j\rangle |a_j\rangle$. So $\langle a_j,\overline{b}_j\rangle$ has modulus one and can thus be absorbed in the choice of the basis so that $b_j=\overline{a}_j$.
\end{proof}
Now we are equipped to prove Prop.\ref{prop:MUB}, which we restate for convenience:
\begin{proposition*}[MUBs] Let $\omega$ be the projector onto a maximally entangled state in $\C^d\otimes\C^d$ and $P:=\1-\omega$.
    There are $(d+1)$ MUBs in $\C^d$ iff there are positive operators $(P_i)_{i=1}^{d+1}$ of rank $(d-1)$ s.t. $P=\sum_{i} P_i$ with partial traces ${\rm tr}_1[P_i]={\rm tr}_2[P_i]\propto \1$ and partial transposes s.t. $\|(P_i+\omega)^{T_1}\|_1=d$. 
\end{proposition*}
\begin{proof}
    If there are $(d+1)$ MUBs, given in terms of unitaries $U_1,\ldots, U_{d+1}$ we can simply take $P_i=R_{U_i}$. These will satisfy Eq.(\ref{eq:RUisum}), and all the other required properties by construction (note that $\|(P_i+\omega)^{T_1}\|_1=d$ is nothing but $(P_i+\omega)^{T_1}\geq 0$).

    Conversely, assume that $P_1,\ldots,P_{d+1}$ exist as claimed. Then each $P_i$ has spectrum in $[0,1]$ (because $P$ has) and at most $(d-1)$ nonzero eigenvalues, so $\tr{P_i}\leq d-1$. However, $d^2-1=\tr{P}=\sum_{i=1}^{d+1} \tr{P_i}$, which means that necessarily $\tr{P_i}=(d-1)$ and each $P_i$ is a Hermitian projector. Then $X_i:=P_i+\omega$ satisfies the assumptions of Lemma \ref{lem:XRU}, which implies that there is a unitary $U_i\in U(d)$ s.t. $P_i=R_{U_i}$. Finally, Lemma \ref{lem:OrthopMUBs} \emph{(2)} completes the proof.
\end{proof}

\bibliographystyle{halpha}
\bibliography{PTI}{}\vspace*{15pt}

\end{document}